\documentclass[aps,pra,superscriptaddress,12pt,longbibliography]{revtex4-2}

\makeatletter
\renewcommand{\p@subsection}{}
\renewcommand{\p@subsubsection}{}
\makeatother

\usepackage{amsmath,amsfonts,amssymb,amsthm}
\usepackage{mathtools}
\usepackage{setspace}
\usepackage[normalem]{ulem} 
\usepackage{stmaryrd}       
\usepackage{dsfont}         
\usepackage{expdlist}
\usepackage{physics}
\usepackage{verbatim}
\usepackage{xspace}
\usepackage{subfigure}

\newcommand{\Tau}{\mathrm{T}}
\newcommand{\RNum}[1]{\uppercase\expandafter{\romannumeral #1\relax}}

\usepackage{graphicx,xcolor}
\usepackage{hyperref}
\hypersetup{pdfstartview=}
\usepackage{url}

\usepackage{verbatim}
\usepackage{alltt}
\usepackage{moreverb}

\usepackage{xr} 

\providecommand{\ignore}[1]{}

\newif\ifcmnt
\cmnttrue
\ifdefined\cmntsoff\cmntfalse\fi

\ifcmnt

\else

\fi

\newcommand{\diag}{\mathrm{Diag}}

\newcommand{\defeq}{\doteq}

\newcommand{\cC}{\mathcal{C}}

\newcommand{\cF}{\mathcal{F}}
\newcommand{\cG}{\mathcal{G}}

\newcommand{\cO}{\mathcal{O}}

\numberwithin{equation}{section}
\newtheorem{theorem}{Theorem}[section]
\newtheorem*{theorem*}{Theorem}

\newtheorem*{problem*}{Problem}
\newtheorem{lemma}[theorem]{Lemma}
\newtheorem*{lemma*}{Lemma}
\newtheorem{corollary}[theorem]{Corollary}

\begin{document}

\title{Multi-mode Gaussian State Analysis with Total Photon Counting}

\author{Arik Avagyan}
\affiliation{National Institute of Standards and Technology, Boulder, Colorado 80305, USA}
\affiliation{Department of Physics, University of Colorado, Boulder, Colorado, 80309, USA}
\author{Emanuel Knill}
\affiliation{National Institute of Standards and Technology, Boulder, Colorado 80305, USA}
\affiliation{Center for Theory of Quantum Matter, University of Colorado, Boulder, Colorado 80309, USA}
\author{Scott Glancy}
\affiliation{National Institute of Standards and Technology, Boulder, Colorado 80305, USA}

\begin{abstract} 
The continuing improvement in the qualities of photon-number-resolving detectors opens new possibilities for measuring quantum states of light. In this work we consider the question of what properties of an arbitrary multimode Gaussian state are determined by a single photon-number-resolving detector that measures total photon number. We find an answer to this question in the ideal case where the exact photon-number probabilities are known. We show that the quantities determined by the total photon number distribution are the spectrum of the covariance matrix, the absolute displacement in each eigenspace of the covariance matrix, and nothing else. In the case of pure Gaussian states, the spectrum determines the squeezing parameters. 
\end{abstract}

\maketitle


\section{introduction}

Gaussian states in continuous variable systems are relatively easy to prepare experimentally and can be used in quantum communication, quantum cryptography, quantum sensing and other areas \cite{wang2007quantum, weedbrook:qc2012a}. In this context, the problem of experimentally analyzing and characterizing Gaussian states becomes important. The most common methods of characterizing single-mode and multi-mode states involve homodyne tomography and variants thereof \cite{weedbrook:qc2012a, paris2003purity,laurat2005entanglement, dauria2005characterization,porzio2007homodyne, rehacek2009effective, dauria2009full, paternostro:qc2009a, buono2010quantum, blandino:qc2012a, esposito2014pulsed}. Homodyne tomography is based on quadrature measurements requiring photodiodes that record the intensity of absorbed light in a high-amplitude regime. More recently, several groups have proposed schemes for characterizing Gaussian states with click detectors or photon-number-resolving (PNR) detectors. 
In particular, Ref.~\cite{fiurasek2004how} proposed a scheme that only uses beam splitters and single-photon detectors to measure the purity, squeezing and entanglement of Gaussian states. Ref.~\cite{wenger2004pulsed} verified this scheme experimentally on single-mode Gaussian states. For single mode states several studies \cite{wallentowitz1996unbalanced, konrad1996direct, manko2003photon,manko2009photon, Leibfried1996exp, nogues2000meas} have shown that the statistics of photon counts obtained after displacing the state by different amounts allows one to estimate the Wigner function. Photon counting with single-mode squeezed states have applications in quantum metrology - in particular, they have been shown to enhance the sensitivity of measuring the phase of an interferometer \cite{wu2019quantum} as well as the coherent displacement of a mechanical oscillator \cite{burd2019quantum}. Numerical \cite{bezerra2021quadrature} studies show that comparatively few number of measurements of the photon number of a single mode non-displaced state allow for accurate determination of its squeezing and temperature parameters. A recent experimental study demonstrated that photon counting allows for a better precision measurement of a weakly-squeezed vacuum state as compared to homodyning \cite{zuo2022determination}. Ref.~\cite{burenkov2017full} showed experimentally that the statistics of two-photon counters can be used to reconstruct the mode structure of a parametric down-conversion source. Full characterization of multi-mode Gaussian states has been shown to be theoretically possible using PNR detectors  on each mode after passing the state through  different linear optical circuits \cite{parthasarathy2015from, kumar2020optimal}. The proposed schemes either yield only general features of an arbitrary multi-mode state like the mean displacement and the determinant of the covariance matrix \cite{fiurasek2004how}, or are difficult to realize in practice \cite{parthasarathy2015from, kumar2020optimal}.

In this light, we consider the problem of what can be learned about an arbitrary multi-mode Gaussian state with a very simple setup - the state is measured by a single PNR detector that returns the total number of photons in all modes. For this study, we make the idealizing assumptions that the detector has no losses nor noise and that we learn the exact probability of observing $n$ photons for every $n$.
We find that the total photon number distribution determines the number of modes that are not in vacuum as well as the spectrum of the covariance matrix in these modes. By spectrum we mean the set of eigenvalues with their multiplicities, if there are degeneracies. The absolute values of the displacement within each eigenspace of the covariance matrix are also determined by the distribution. Conversely, any two states
with the same covariance matrix spectrum and absolute eigenspace displacements
have the same photon number statistics. If the state is pure, this implies that the distribution determines the squeezing spectrum, by which we mean the set of squeezing parameters, as well as the absolute displacement along each squeezing axis (or subspace when the given squeezing value is degenerate).
We discuss the interpretation of the covariance matrix spectrum for mixed states and identify representatives with diagonal covariance
matrices for each equivalence class of Gaussian states with identical photon number statistics.

The paper is organized as follows. In Sect.~\ref{sec:problemFormulation} we formulate the problem. Our main result is a parametrization theorem that shows that for Gaussian states, the total photon number distributions are bijectively parametrized by the covariance matrix spectrum
and the absolute displacements in the covariance matrix eigenspaces.
The parametrization theorem is established in Sect.~\ref{sec:paramTheorem}. Our main tool is the Husimi representation, from which we compute
the expectations of the anti-normally ordered powers of the total number operators. Then, we use generating functions to prove the parametrization theorem. In Sect.~\ref{sec:normalParams} we 
show that for pure Gaussian states, the covariance matrix spectrum determines
the squeezing spectrum. For mixed Gaussian states we determine the set of
squeezing spectra of pure states from which the state being measured could be obtained
by adding Gaussian displacement noise. We also study the set of diagonal covariance matrices with the same spectrum. We conclude in Sect.~\ref{sec:conclusion}.

\section{Problem Formulation}\label{sec:problemFormulation}

We assume familiarity with quantum optics mode operators and
  phase space representations. See~\cite{leonhardt:qc1997a} for a
  pedagogical treatment.

  Consider $S$ modes characterized by annihilation operators
  $\hat{a_{i}}$ for $i=1,\ldots, S$. Our analysis does not depend on
  the particular physical realization of these modes, but we refer to
  the excitations of the modes as photons. The canonical quadrature
  operators $\hat{q_{i}}$ and $\hat{p_{i}}$ of mode $i$ are defined so
  that \(\hat{a}_i = \frac{1}{\sqrt{2}} (\hat{q_i}+i\hat{p_i})\) and
  \(\hat{a}_i^\dagger = \frac{1}{\sqrt{2}} (\hat{q_i}-i\hat{p_i})\).
  We define the vector of quadrature operators as \(\vec{\hat{r}} =
  (\hat{q}_1,\hat{p}_1,\hdots, \hat{q}_S, \hat{p}_S) \). We further define the vectors of annihilation and creation operators as $\vec{\hat{a}} = (\hat{a_1},\hdots,\hat{a_S})$ and $\vec{{\hat a}}^\dagger = (\hat{a_1}^\dagger,\hdots,\hat{a_S}^\dagger)$, respectively. We use the
  convention that operators are denoted with ``hats''. Density operators are excluded from this convention.
  The variables without the hats denote scalar values or vectors of values. For
  example, $\vec{r}$ is a vector of $2S$ values, which we interpret as
  values of phase-space coordinates.

  We assume that the state $\rho$ being measured is Gaussian, characterized
  by a displacement $\vec{d}$ with entries $d_i=\langle \hat{r}_{i} \rangle$, and a
  covariance matrix $\Gamma$ with entries $\Gamma_{ij} = \langle
  \hat{r}_i \hat{r}_j + \hat{r}_j \hat{r}_i \rangle - 2 \langle
  \hat{r}_i \rangle \langle \hat{r}_j \rangle$. 
The notation $\langle \ldots \rangle$ denotes
  expectation with respect to $\rho$. With these conventions the covariance matrix of the vacuum state corresponds to the identity. The Wigner function of $\rho$ is given by
  \begin{equation}\label{eqtn 1}
     W(\vec{r}) = \frac{1}{\pi^S} \frac{1}{\sqrt{\det(\Gamma)}} e^{-(\vec{r}-\vec{d})^T \Gamma^{-1}  (\vec{r}-\vec{d})}.
  \end{equation}
This corresponds to Eq.~20 in Sec.~$\RNum{2} A$ of \cite{weedbrook:qc2012a}, where this expression is derived. We note that we use a different convention than Ref.~\cite{weedbrook:qc2012a}. The annihilation operators are defined as \(\hat{a}_i = \frac{1}{2} (\hat{q_i}+i\hat{p_i})\) there, and the covariance matrix is equal to $\frac{1}{2} \Gamma$, which results in a different form of the Wigner function.

  We use the Husimi representation \cite{leonhardt:qc1997a}[Ch.~3] of Gaussian states. The Husimi representation can be obtained from the
  Wigner function by convolution with the Gaussian $\frac{1}{\pi^S} e^{- \abs{\vec{r}}^2} $ and is therefore also Gaussian. With our conventions, the Husimi representation of $\rho$ is 
  \begin{equation}\label{eqtn 2}
    Q(\vec{r}) =  \frac{1}{\pi^S} \frac{1}{\sqrt{\det(I+\Gamma)}} e^{-(\vec{r}-\vec d)^T (\Gamma +I)^{-1}  (\vec{r}-\vec d)},
  \end{equation} 
  where $I$ is the $2S \times 2S$ identity matrix. 

  We analyze the situation where the $S$ modes are measured by an
  ideal PNR detector that does not distinguish the modes. The detector's output is
  the number $n$ of photons observed, and the associated operator is
  the projector $\hat{\Pi}_{n}$ onto the subspace of states with $n$
  photons. Let $\hat n$ denote the total photon number operator on the
  $S$ modes so that $\hat n =
  \sum_{i=1}^{S}\hat{a_{i}}^{\dagger}\hat{a_{i}}$. For this work we
  assume that we have learned the exact probabilites $\langle
\hat{\Pi}_{n} \rangle$ of having $n$ photons for every $n$. This means that
  we can assume as given the photon number distribution. In the next
  section, we solve the following problem: 

  \begin{problem*}What features of $S$-mode Gaussian states are determined by their photon number distribution?
  \end{problem*}

  To help solve the problem, we need a few observations about
  relationships between the photon number distribution, its moments,
  and the expectations of the anti-normally ordered powers of the total
  photon number operator. The moments of the photon number distribution are given by $\langle \hat{n}^{k}\rangle$. Our calculations are
  simplified by considering instead the anti-normally ordered moments
  given by \(\langle \vdots \hat{n}^l\vdots \rangle \).  Here, the
  vertical triple dots denote anti-normal ordering of all mode
  operators in the formal expression between the triple dots - that is, all creation operators are moved to the right of the annihilation operators.

The Husimi representation of states satisfies the optical
    equivalence theorem for anti-normal order. A general treatment of
    this theorem is in Refs. \cite{agarwal1970calculus1, agarwal1970calculus2}. These references
    explain orderings for one mode. To generalize the treatment to
    multiple modes, it suffices to apply the fact that operators from
    different modes commute. For the Husimi representation, it
    implies that expectations of expressions $f(\vec{\hat a},
    \vec{{\hat a}}^\dagger)$ whose terms are already in anti-normal order can be
    evaluated as the expectation of $f(\vec{\alpha},\vec{\bar\alpha})$
    with respect to the Husimi representation of the state. That is,
    if $Q(\vec{\alpha},\vec{\bar\alpha})$ is the Husimi representation
    of $\rho$, then
    \begin{align}
      \langle f(\vec{\hat a}, \vec{{\hat a}}^\dagger ) \rangle_\rho &=
      \langle f(\vec{\alpha}, \vec{\bar \alpha} ) \rangle_H =
         \int \prod_i d\alpha_i d\bar\alpha_i f(\vec{\alpha},\vec{\bar\alpha}) Q(\vec{\alpha},\vec{\bar\alpha}).
    \end{align}
    It follows that anti-normal ordering, like other such orderings, has the property
    that for a general expression $g(\vec{\hat a}, \vec{{\hat a}}^\dagger)$ that may include terms that are not in anti-normal order, we have
    $\langle \vdots g(\vec{\hat a}, \vec{{\hat a}}^\dagger)\vdots\rangle_\rho  =
    \langle g(\vec{\alpha},\vec{\bar\alpha})\rangle_H$ - see Thm. I of Ref.~\cite{agarwal1970calculus1}.
We utilize the expectations of the anti-normally ordered powers of the number operator. The $l$'th power of the total number operator is expressed as ${\hat n}^l=(\sum_{i=1}^S \hat{a_i} {\hat{a_i}}^\dagger)^l$.
    In phase space, $\alpha_i\bar\alpha_i = (p_i^2+q_i^2)/2$, and as a result we obtain
    \begin{align}\label{eqtn 3}
      \langle  \vdots \hat{n}^l \vdots \rangle &= \langle  \vdots (\sum_{i=1}^S \hat{a}_i^\dagger \hat{a}_i)^l \vdots \rangle \nonumber \\
      & = \int \prod_id\alpha_id\bar\alpha_i Q(\vec{\alpha},\vec{\bar\alpha}) (\sum_i\alpha_i\bar\alpha_i)^l \nonumber \\
      & = \frac{1}{2^l} \int d^{2S}\vec{r}  Q(\vec{r}) \abs{\vec{r}}^{2l}. 
    \end{align}

Notice that for Gaussian states these integrals converge. Below we show that the anti-normally ordered moments determine and are determined by the $\langle \hat{n}^{k}\rangle$. This implies that the latter are defined as well. 

  Our result relies on analyzing the generating function for anti-normally ordered moments:
  \begin{align}
       G(z) &= \langle \vdots e^{-z \hat n}\vdots\rangle\nonumber\\
         &=\sum_{j} (-1)^{j}\frac{1}{j!} \langle \vdots \hat n^{j}\vdots\rangle z^{j}.
         \label{eqtn Gz}
  \end{align}
For Gaussian states, we find that $G(z)$ is analytic in a neighborhood
  of $z=0$, see Eqs.~\ref{eq:Gintegral2} and ~\ref{eq:Gintegral3} in the next section. The next paragraph shows that the anti-normally ordered moments carry the
  same information as the usual moments $\langle \hat{n}^k \rangle$. The relationship between the two types of moments implies that the generating function for the $\langle \hat{n}^k \rangle$ is also analytic in a neighborhood of $z=0$, where it is determined by $G(z)$. Crucially, this means that for Gaussian states $G(z)$ determines the total photon number distribution~\cite[Ch.~4]{severini2005elements}.

According to
  Ref.~\cite{shalitin1979trans}, for one mode ($S=1$), the operator-valued
  generating functions $e^{-x\hat{n}}$ and $\vdots e^{-x \hat
      n}\vdots$ are related by
  \begin{equation}\label{usefuleqtn 1}
    e^{-x \hat{n}} = e^x  \vdots  e^{(1-e^{x}) \hat{n}}  \vdots\;.
  \end{equation}
  This identity extends to the total photon number in an arbitrary numbers of modes as follows.
  \begin{align}
    e^{-x \sum_{i=1}^S \hat{n}_i} = \prod_{i=1}^S e^{- x \hat{n}_i}  
       &= \prod_{i=1}^S \left[ e^x  \vdots  e^{(1-e^{x}) \hat{n}_i}  \vdots \right] \nonumber\\
       &=\left[  \prod_{i=1}^S e^x \right]   \vdots  e^{(1-e^{x}) \sum_{i=1}^S \hat{n}_i}  \vdots\nonumber\\
       &=e^{xS}\vdots  e^{(1-e^{x}) \sum_{i=1}^S \hat{n}_i} \vdots,
    \label{usefuleqtn 2}
  \end{align}
where the second line is obtained from the first by applying the observation that antinormal ordering and products over distinct, commuting modes can be interchanged without changing the resulting operator.
Introduce the new variable $z$ satisfying $x=\ln(1+z)$ and substitute in Eq.~\ref{usefuleqtn 2} to obtain
  \begin{equation}\label{usefuleqtn 3}
      e^{- \ln(1+z)  \hat{n}} = (1+z)^S   \vdots  e^{- z \hat{n}}  \vdots.
  \end{equation}
  From this we can write
  \begin{align}\label{usefuleqtn 4}
      G(z) = \langle \vdots  e^{- z \hat{n}}  \vdots \rangle = \langle  e^{- \ln(1+z)  \hat{n}}\rangle/(1+z)^{S}.
  \end{align}
  Both expressions for $G(z)$ are well-defined as generating functions.
This identity implies that the anti-normally ordered moments of degree $j$
  are a linear combination of the usual moments of degree at most $j$ and vice-versa. This can be verified as follows: Eq.~\ref{eqtn Gz} shows that the coefficient of $j$'th power of $z$ in $\langle \vdots  e^{- z \hat{n}}  \vdots \rangle$ is proportional to $\langle \vdots \hat{n}^j \vdots \rangle$. Expanding the rightmost expression in Eq.~\ref{usefuleqtn 4} one can see that the coefficient of $z^j$ is a linear combination of the $\langle \hat{n}^k \rangle$ for $k \leq j$. To see the reverse, multiply both expressions for $G(z)$ by $(1+z)^S$ and substitute $z = e^x -1$. Then, the rightmost expression has the coefficients proportional to the powers of $\langle \hat{n}^j \rangle$, and for each $j$ the transformed expression on the left can be seen to be a linear combination of the $\langle \vdots \hat{n}^k \vdots \rangle$ for $k \leq j$.

\section{The Parametrization Theorem}\label{sec:paramTheorem} 

In this section, we prove the parameterization theorem. We first show
that the anti-normally ordered generating function $G(z)$ defined in
  Eq.~\ref{eqtn Gz} may be expressed as a Gaussian integral by means of the
  expression for the anti-normally ordered moments in terms of the
  Husimi representation in Eq.~\ref{eqtn 3}. Further, we show that $G(z)$ is analytic in a neighborhood of $z=0$. As explained in Sect.~\ref{sec:problemFormulation}, this implies that $G(z)$ determines the total photon number distribution.
We find that $G(z)$ only depends on the spectrum of the state's covariance matrix $\Gamma$ and
the absolute displacement of the state within the eigenspaces of $\Gamma$.
According to the parametrization theorem the reverse also holds, that
is, these parameters are determined by $G(z)$. Let
$\{\lambda_{i}\}_{i=1}^{N}$ be the distinct eigenvalues of $\Gamma$ in
decreasing order. Let $V_{i}$ be the eigenspace of $\Gamma$ for
eigenvalue $\lambda_{i}$ and $k_{i}$ the dimension of $V_{i}$. The
displacement $\vec d$ can be written uniquely as a sum
$\sum_{i=1}^{N}\vec d_{i}$ with $\vec d_{i}\in V_{i}$. Let $d_{i} =
|\vec d_{i}|$. We refer to the family $\{(\lambda_{i}, k_{i},
d_{i})\}_{i=1}^{N}$ as the ``normal parameters'' of the Gaussian
state.

  \begin{theorem}[Parametrization Theorem]\label{theorem 1}
    The total photon number distribution of a Gaussian state determines and is determined by the normal parameters of the state.
  \end{theorem}

  \begin{proof}
      
     As explained in Sect.~\ref{sec:problemFormulation}, the total photon number
    distribution determines the anti-normally ordered generating
    function $G(z)$. We show that $G(z)$ is analytic in a neighborhood of the origin, which implies that $G(z)$ determines the total photon number distribution. This implies an equivalence between the total photon number distribution and $G(z)$. We further show that $G(z)$ determines and is determined by the normal parameters. The theorem statement then follows from these two equivalences.

    In terms of the Husimi representation 
    $Q(\vec{r})$ and in consideration of Eq.~\ref{eqtn 3},  $G(z)$ is expressed as
    \begin{align}
        G(z) & =   \sum_{l=0}^\infty \frac{(-z)^l}{l!} \langle   \vdots  \hat{n}^l  \vdots \rangle\nonumber\\ 
                &=  \sum_{l=0}^\infty \frac{(-z)^l}{2^l l!} \int dr^{2S} Q(\vec{r}) r^{2l}\nonumber \\
                & = \int dr^{2S} Q(\vec{r}) e^{-\frac{1}{2} z r^2},
       \label{eq:t1prf1}
    \end{align}
    where we used the convention that $r=|\vec{r}|$. According to Eq.~\ref{eqtn 2} $Q(\vec{r})$ is a
  Gaussian with covariance matrix $(\Gamma+I)/2$ and
    displacement $\vec d$. Thus, the integral in Eq.~\ref{eq:t1prf1} is a Gaussian integral that can
   be evaluated to obtain a closed form expression for $G(z)$. See,
   for example, Ref.~\cite[Ch.~4]{mathai1992quadratic}. To simplify
   the expressions, we define $\Gamma' = (\Gamma+I)^{-1}$ and $\vec{y} = \vec{r}-\vec{d}$.
   The evaluation
   goes as follows:
    \begin{align}
      G(z) & = \int dr^{2S} Q(\vec{r}) e^{-\frac{1}{2} z r^2}  \nonumber\\
      & =  \frac{ \sqrt{\det(\Gamma')}}{\pi^S}  \int dr^{2S} e^{-(\vec{r}-\vec{d})^T \Gamma'  (\vec{r}-\vec{d})} e^{-\frac{1}{2} z r^2}   \nonumber\\ 
      &= \frac{ \sqrt{\det(\Gamma')}}{\pi^S}  \int dr^{2S} e^{-(\vec{r}-\vec{d})^T \Gamma'  (\vec{r}-\vec{d})} e^{-\frac{1}{2} z (\vec{r}-\vec{d})^{T}(\vec{r}-\vec{d}) - z (\vec{r}-\vec{d})^{T}\vec{d} - \frac{1}{2} z \vec{d}^{T}\vec{d}}\nonumber\\
      &= \frac{ \sqrt{\det(\Gamma')}}{\pi^S}  \int dy^{2S} e^{-\vec{y}^T \Gamma'  \vec{y}} e^{-\frac{1}{2} z \vec{y}^{T}\vec{y} - z \vec{y}^{T}\vec{d} - \frac{1}{2} z \vec{d}^{T}\vec{d}}\nonumber\\
      &= \frac{ \sqrt{\det(\Gamma')}}{\pi^S}  e^{- \frac{1}{2} z \vec{d}^{T}\vec{d}}\int dy^{2S} e^{-\vec{y}^T (\Gamma'  +\frac{1}{2} z I)\vec{y}- z \vec{y}^{T}\vec{d}}. \label{eq:Gintegral}
    \end{align}
 That $G(z)$ is determined by the normal parameters can be
      deduced from the last expression. First, $\det(\Gamma')$ depends
      only on the eigenvalues of $\Gamma'$, and these eigenvalues are
      derived from the normal parameters as $1/(\lambda_i +1)$ with
      multiplicity $k_i$. Second, we can change variables in the
      integral to diagonalize $\Gamma'$ and standardize $\vec{d}$. Let
      $O$ be an orthogonal matrix for which $O\Gamma' O^T$ is diagonal
      with the eigenvalues in non-ascending order on the diagonal and
      $O^T \vec{d}$ has the property that its nonzero entries are
      non-negative and associated with the first coordinate of each
      eigenspace block of $O\Gamma' O^T$ with the same
      eigenvalue. To achieve the latter property, it suffices to choose appropriate
      orthogonal transformations within each eigenspace block.
      Then $O\Gamma' O^T$ is determined by the $\lambda_i$
      and $k_i$, and the normal parameter $d_i$ is the nonzero entry of $O^T\vec{d}$ associated with
      the eigenspace block for eigenvalue $\lambda_i'= 1/(\lambda_i+1)$ of
      $O\Gamma' O^T$. Changing variables
      according to $\vec{ \tilde{y}} = O \vec y$ is equivalent to replacing
      $\Gamma'$ with $O\Gamma' O^T$ and $\vec{d}$ by $O^T\vec{d}$. This
      equivalence hinges on the rotational invariance of the measure of integration
      $dy^{2S}$. After this transformation, the integral factors as a product
      over each coordinate separately and we find that the value of the integral is determined
      by the normal parameters. Here is the explicit calculation. To express this transformation in the integral we index the new variable of integration $\vec{\tilde{y}}$ according to the eigenspace blocks as
        $\tilde{y}_{ij}$, where $i$ indicates the $i$'th of $N$ blocks and $j$ indicates the $j$'th of $k_i$ coordinates in the block. Then
    \begin{align}\label{eq:Gintegral2}
    G(z) &= \frac{ \sqrt{ \prod_{i=1}^N (\lambda_i')^{k_i} }}{\pi^S}  e^{- \frac{1}{2} z \sum_{i=1}^N d_i^2}\int d \tilde{y}^{2S} e^{ - \sum_{i=1}^N \left[(\lambda_i'+\frac{1}{2}z) \sum_{j=1}^{k_i}\tilde{y}_{ij}^2\right]  - \sum_{i=1}^N z \tilde{y}_{i1} d_i}
    \end{align} 
 The integrand in Eq.~\ref{eq:Gintegral2} factors as a product of exponentials, each a function of one of the coordinates $\tilde{y}_{ij}$. Therefore, the integral can be expressed as a product of one-dimensional Gaussian integrals. Provided $\frac{1}{2}z > - \textrm{min} \{ \lambda_i'  \}$, every integral in the product is finite, so $G(z)$ evaluates to
\begin{align}\label{eq:Gintegral3}
    G(z) &= \frac{ \sqrt{ \prod_{i=1}^N (\lambda_i')^{k_i} }}{\pi^S}  e^{- \frac{1}{2} z \sum_{i=1}^N d_i^2}  \prod_{i=1}^N \left[\sqrt{\frac{\pi}{ \lambda_i'+\frac{1}{2}z }}\right]^{k_i} e^{\frac{z^2d_i^2}{4(\lambda_i'+\frac{1}{2}z)} } \nonumber \\
    & =  \prod_{i=1}^N \left[\frac{ \lambda_i' }{\lambda_i'+\frac{1}{2}z} \right]^{k_i/2}  e^{- \frac{1}{2} z d_i^2+\frac{z^2d_i^2}{4(\lambda_i'+\frac{1}{2}z)} } \nonumber \\
    &= \prod_{i=1}^N \left[\frac{ \lambda_i' }{\lambda_i'+\frac{1}{2}z} \right]^{k_i/2}  e^{-\frac{-z d_i^2 \lambda_i'}{2(\lambda_i'+\frac{1}{2}z)} } .
\end{align} 
    Since $\lambda_i >0$ the minimum of the $\lambda_i'$ is a strictly positive number. Therefore, $G(z)$ is analytic in a neighborhood of the origin.

One way to obtain the normal parameters from $G(z)$ is to look at the first derivative of its natural logarithm. $\ln(G(z))$ is a multi-valued function, where the different ``branches'' differ by an additive constant. Thus, the derivative of $\ln(G(z))$ is a well-defined, single valued function for $z>-2\min\{\lambda_i'\}$: 
    \begin{align}
       L(z) &\defeq \frac{d \ln(G(z))}{dz}  
           \nonumber\\
          &= 
          -\frac{1}{2}\sum_{i}\frac{d}{dz}k_{i} \ln(\lambda_{i}' + \frac{1}{2}z)
          -\frac{1}{2}\sum_{i}\frac{d}{dz} z d_{i}^{2} \lambda_{i}'(\lambda_{i}'+\frac{1}{2}z)^{-1}
          \nonumber\\
          &= -\frac{1}{4}\sum_{i} k_{i}(\lambda_{i}'+\frac{1}{2}z)^{-1}
                -\frac{1}{2}\sum_{i} d_{i}^{2} \lambda_{i}' (\lambda_{i}' +\frac{1}{2}z)^{-1}
                +\frac{1}{4}\sum_{i}zd_{i}^{2} \lambda_{i}' (\lambda_{i}'+\frac{1}{2}z)^{-2}\nonumber\\
           &= -\sum_{i}  \left(\frac{k_{i}}{2} (2\lambda_{i}' + z)^{-1}
                              + 2d_{i}^{2} (\lambda_{i}')^2(2\lambda_{i}'+z)^{-2}\right).
           \label{eq:ddzlogg}
    \end{align}
 Because the $G(z)$ is defined in a neighborhood of the origin, $L(z)$ is also defined in a neighborhood of the origin, but it can be extended to a function with a maximal domain of definition on the complex plane.  This is true for any function defined on a non-empty, open subset of the complex plane, and the procedure is known as analytic continuation \cite[Ch.~16]{rudin1987real}. In this case the analytic continuation of $L(z)$ is the extension of the domain to all $z$ where the expression on the right side of Eq.~\ref{eq:ddzlogg} is defined. We refer to the analytic continuation of $L(z)$ as $L_{\text{a.c.}}(z)$. According to Eq.~\ref{eq:ddzlogg}, $L_{\text{a.c.}}(z)$ is analytic except at poles of at most second order at $z_i=-2\lambda_i'$ for each $i$.  By uniqueness of analytic continuations, the locations of the poles and their coefficients are determined by $G(z)$.  The positions $z_{i}$ of the poles and the coefficients of the corresponding orders $1/(z-z_{i})$ and $1/(z-z_{i})^{2}$ can in principle be extracted by contour integration.  The position of each pole determines a $\lambda_{i}'$ and therefore a $\lambda_{i}$.  From the coefficient of the pole at $z_i$ we obtain the multiplicity parameter $k_{i}$ and the displacement $d_{i}$. It follows that $G(z)$ determines the normal parameters.
\end{proof}

\section{Interpretation of Normal Parameters for Pure and Mixed States}\label{sec:normalParams}

Thm.~\ref{theorem 1} shows that the total photon number
distribution determines the spectrum of the covariance matrix
and the absolute displacement in each eigenspace, and nothing
else. But what does this tell us about the physical properties of the
Gaussian state such as the amount of squeezing along different
directions of phase-space or the temperature parameters (defined in the paragraph below) of the modes? The set of all Gaussian states can be divided into equivalence
classes, such that the states in a given equivalence class have the
same normal parameters. We are interested in
characterizing the physical properties of the states that belong to
the same equivalence class. We show that for pure Gaussian
  states, the squeezing parameters are determined by the normal
  parameters. We characterize the set of normal
  parameters of Gaussian states, and we show that for such normal parameters, there is
  always a Gaussian state with diagonal covariance matrix in a fixed
  mode basis with these normal parameters. We further investigate sets of states that have the same normal parameters and whose covariances are all diagonal in the same mode basis. The background material
for this section can be found in reviews and
textbooks such as Ref.~\cite{serafini2017quantum}. 

Let $\Gamma$ be the $2S\times 2S$ covariance matrix of the observed
state in some mode basis. The mode basis determines an antisymmetric matrix $J$ that is preserved by the action of Gaussian unitaries on the mode operators. We order the coordinates so
that the antisymmetric matrix $J$ is block diagonal with $S$ blocks of
the form $\begin{pmatrix} 0& 1\\-1&0\end{pmatrix}$.  Gaussian
  unitaries that involve no displacement are characterized by
  transformation matrices $A$ that satisfy $A^{T}J A = J$. Such matrices are called symplectic. There exists a symplectic
matrix $A$ such that $\Gamma = A^{T}\Tau A$ \cite{serafini2017quantum}[Secs.~3.2.3 and ~3.2.4] with $\Tau$
diagonal and consisting of $S$ blocks of the form $\begin{pmatrix}\nu_{i} & 0 \\
  0 & \nu_{i}\end{pmatrix}$ with $\nu_{i}\geq 1$, where we normalized
mode operators so that the vacuum covariance matrix is the
identity. We refer to $\Tau$ as the symplectic diagonalization of $\Gamma$, and to the family consisting of the $\nu_i$ as the symplectic spectrum of $\Gamma$. The Gaussian state with covariance matrix $\Tau$ is thermal in each mode, and the modes are uncorrelated. We call such states ``independently thermal states'', where $\nu_{i}$ is the temperature parameter
for the $i$'th mode. In terms of the expected number of
quanta in mode $i$, the temperature parameter $\nu_{i}$ is given by
$\nu_{i}=2 \langle \hat{n}_i \rangle +1$.  Symplectic transformations
can be physically realized by a combination of squeezing and 
linear optical transformations.  Passive linear optical
transformations are represented by symplectic matrices $O$ that are
also orthogonal, that is $O^{T}O=I$.  Every symplectic matrix has a
representation $A=K Q L$ where $K$ and $L$ are symplectic and
orthogonal, and $Q$ squeezes each mode by different amounts \cite{serafini2017quantum}[Sec.~5.1.2] . Such a
$Q$ is diagonal with diagonal blocks of the form $\begin{pmatrix}
  e^{r_i} & 0 \\ 0 & e^{-r_i} \end{pmatrix} $ where $r_{i}$ is the
squeezing parameter for mode $i$.  The squeezing parameters of $\Gamma$ are determined by $Q$. For one mode, $S=1$, the symplectic diagonalization $\Tau$ is proportional to the
identity and commutes with $K$. Consequently $\Gamma=L^TQ\Tau QL$, where $Q\Tau Q$ has spectrum $(\nu_1e^{2r_1},\nu_1e^{-2r_1})$, and therefore, so does $\Gamma$. 
In this case, the thermal and squeezing parameters are determined by the spectrum of $\Gamma$.  

For a multimode state where $\Gamma$ is diagonal, the previous paragraph implies that $\Gamma$ is composed of $S$ consecutive blocks of the form $\begin{pmatrix}\nu_ie^{2r_i}&0\\0&\nu_ie^{-2r_i}\end{pmatrix}$, where $\nu_i$ and $r_i$ are the temperature and squeezing parameters of mode $i$, respectively. The product of the diagonal elements of these $2\times 2$ matrices are $\nu_i^2$,
which satisfy $\nu_i^2 \geq 1$.
Conversely, consider any diagonal positive matrix $M$ with $S$ diagonal
  $2\times 2$ blocks, where the block for mode $i$ is of the form
  $\begin{pmatrix}\gamma_{i}&0\\0&\gamma_{i}'\end{pmatrix}$ with
  $\gamma_{i}\gamma_{i}' \geq 1$. Then $M$ is the covariance matrix of
  a Gaussian state. To see this it suffices to transform for each $i$,
  the $i$'th mode's block with the symplectic diagonalization $2\times 2$
  matrix
  $B_{i}=\begin{pmatrix}(\gamma_{i}'/\gamma_{i})^{1/4}&0\\0&(\gamma_{i}/\gamma_{i}')^{1/4}\end{pmatrix}$. This
  gives a covariance matrix that is independently thermal in each mode
  as described above.  We say that $M$ is the covariance matrix of a
  Gaussian state where the $i$'th mode is a squeezed thermal state.
  The $i$'th mode has temperature parameter
  $\nu_{i}=\sqrt{\gamma_{i}\gamma_{i}'}$ and squeezing parameter
  $r_{i}=\ln(\gamma_{i}/\nu_{i})/2$.

We call covariance matrices of Gaussian pure states ``pure covariance matrices''. To determine these states' parameters, we need the following characterizations of pure covariance matrices:
\begin{lemma}\label{lem:covbasics}
  Let $\Gamma$ be a covariance matrix of a Gaussian state on $S$
  modes.  The following are equivalent: 1. The matrix $\Gamma$ is
  pure.  2. $\det(\Gamma)=1$.  3. The eigenvalues
  $(\gamma_{j})_{j=1}^{2S}$ in non-ascending order of $\Gamma$ satisfy
  the tight pairing condition $\gamma_{j}\gamma_{2S+1-j}=1$. Furthermore, in case
  3.~the quantities $\ln(\gamma_{j})/2$ for $j\leq S$ are the
  squeezing parameters of the state.
\end{lemma}
\begin{proof}
  The equivalence of 1. and 2. can be found in \cite[Ch. 3, Sec. 5]{serafini2017quantum}, but we provide a proof for completeness.
  Since displacements are realized unitarily and do not affect the covariance matrix, we may assume that the state
  is undisplaced so that the quadrature operators have zero mean.
  Write $\Gamma=A^{T}\Tau A$ with $A$ simplectic and $\Tau$ diagonal
  with thermal blocks. Since $A$ is realized by a Gaussian unitary transformation,
  $\Gamma$ is pure iff $\Tau$ is.
  The covariance matrix $\Tau$ is pure iff $\Tau=I$, or equivalently, iff the temperature parameters of all modes are $0$. The identity $A^{T} J A = J$ implies that
  $\det(A)=\pm 1$. Thus $\det(\Gamma)= \det(\Tau)$. The form of $\Tau$
  implies that $\det(\Tau) \geq 1$ with $\det(\Tau)=1$ iff all temperature parameters
  are zero, that is, iff $\Tau$ is the covariance matrix of vacuum.
  This proves the equivalence of 1. and 2. 

  Write $A=K Q L$ with $K$ and $L$ simplectic orthogonal and $Q$
  diagonal with blocks of the form $\begin{pmatrix} e^{r_j} & 0
    \\ 0 & e^{-r_j} \end{pmatrix}$. We may assume without loss of
  generality that $r_{j}\geq 0$. According to the previous
  paragraph, if $\Gamma$ is pure, then $\Tau=I$.  Since
  $K^{T}K=I$, we have $\Gamma = A^{T}\Tau A = L^{T} Q^{T} K^{T} \Tau KQL = L^{T}Q^{2}L$. Since $L$ is
  orthogonal, the spectrum of $\Gamma$ is that of $Q^{2}$, and the
  pairing condition is satisfied by $Q^{2}$. The relationship of
  the eigenvalues to the squeezing parameters is implied by this
  form.  Conversely, suppose that the pairing condition is
  satisfied by $\Gamma$. Then $\det(\Gamma)=1$ so $\Gamma$ is
  pure.
\end{proof}

The next theorem establishes the relationship between normal parameters
and squeezing parameters of pure Gaussian states.

\begin{theorem}\label{thm:pure}
  Let $\cF=\{(\lambda_{i},k_{i},d_{i})\}_{i=1}^{N}$ be the family of
  normal parameters of a Gaussian state. The state is pure iff
  $\prod_{i}\lambda_{i}^{k_{i}}=1$. For pure states, the squeezing
  parameters are determined as follows: Let $(\gamma_j)_{j=1}^{2S}$ be the
  non-ascending sequence of length $2S=\sum_{i}k_{i}$ in which
  $\lambda_{i}$ occurs $k_{i}$ times. The $S$ squeezing parameters
  of the state are given by $\ln(\gamma_{j})/2$ for $j=1,\ldots,
  S$. 
\end{theorem}

\begin{proof}
  Let $\Gamma$ be the covariance matrix for the Gaussian state.
  For the first statement, it suffices to observe that $\prod_{i}\lambda_{i}^{k_{i}}$
  is the determinant of $\Gamma$ and apply Lem.~\ref{lem:covbasics}.
  If the state is pure, write $\Gamma = L^{T}QK^{T} I K Q L$ as discussed
  at the beginning of this section. Since $K$ is orthogonal, $K^{T}IK = I$
  and $\Gamma = L^{T}Q^{2}L$. Since $L$ is orthogonal, the spectrum
  of $\Gamma$ is the spectrum of $Q^{2}$, whose entries are
  $e^{\pm 2 r_{j}}$, where the $r_{j}$ are the squeezing parameters.
  This proves the second statement.
\end{proof}

As noted at the beginning of the section, for one mode, the squeezing and the thermal
parameters of a Gaussian state are determined by the spectrum. Therefore, in this case the normal parameters determine the temperature and squeezing parameter of the state. If the state is
  unsqueezed, one can determine the absolute displacement. Otherwise, one can determine the absolute displacements in the squeezed and in the antisqueezed directions.
For mixed Gaussian states on two or more modes, it is in
general not possible to determine the squeezing and thermal
parameters from the normal parameters, but we can determine
diagonal representatives of the set of Gaussian states with the
same normal parameters and characterize the set of normal parameters.

\begin{lemma} \label{lem:exdiagcov}
  Let $\Gamma$ be the covariance matrix of a Gaussian
  state.  Then there exists a diagonal covariance matrix $D$ of a Gaussian state with the same spectrum as $\Gamma$.
\end{lemma}

\begin{proof}
  To prove the lemma we use the fact that there exists a pure state
  covariance matrix $\Gamma_p$ such that $\Gamma_p\leq \Gamma$ and
  then apply Weyl's monotonicity principle
  \cite[Ch.~3]{horn1991topics} to compare the spectra. For the first
  step, we write $\Gamma = A^{T} \Tau A$ with $A$ simplectic and
  $\Tau$ independently thermal in each mode. Then $\Tau\geq I$, and
  I is the covariance matrix of vacuum, which is a pure Gaussian
  state.  Therefore, $\Gamma_p = A^{T}I A$ is the covariance
  matrix of a pure state, and $\Gamma = A^{T}\Tau A \geq A^{T}I A =
  \Gamma_p$.  Let $(\gamma_{j})_{j=1}^{2S}$ and
  $(\gamma_{j}')_{j=1}^{2S}$ be the eigenvalues of $\Gamma$ and
  $\Gamma_p$ in non-ascending order.  By Weyl's monotonicity
  principle, $\gamma_{j}\geq \gamma_{j}'$.  By
  Lem.~\ref{lem:covbasics}, for $j\leq S$ we have
  $\gamma_{j}'\gamma_{2S+1-j}'=1$, which implies that
  $\gamma_{j}\gamma_{2S+1-j}\geq 1$.  Let $D$ be the $2S\times 2S$
  diagonal matrix where the $j$'th mode's $2\times 2$ block has
  diagonal $(\gamma_{j},\gamma_{2S+1-j})$.  Then $D$ has the same
  spectrum as $\Gamma$ and as observed at the beginning of this
  section, $D$ is the covariance matrix of a Gaussian state.
  
\end{proof}

\begin{corollary}\label{s4cor1}
  Consider the family of normal parameters
  $\cF=\{(\lambda_{i},k_{i},d_{i})\}_{i=1}^{N}$ for $S$ modes. Let $(\gamma_j)_{j=1}^{2S}$
  be the non-ascending sequence of length $2S = \sum_{i}k_{i}$ in which
  $\lambda_{i}$ occurs $k_{i}$ times. 
  There exists a displaced Gaussian state with diagonal covariance matrix whose family of normal parameters is $\cF$.
\end{corollary}

\begin{proof}
  By Lem.~\ref{lem:exdiagcov} there exists a diagonal covariance
  matrix of a Gaussian state $\rho'$, whose non-zero
  entries are composed of the $\gamma_i$.  We may
  assume that $\rho'$ is undisplaced, so that it has zero-mean
  quadratures. To obtain the desired Gaussian state, it suffices
  to displace the quadratures associated with the first coordinate
  of each set of coordinates with identical eigenvalues by
  $d_{i}$.
\end{proof}

\begin{theorem}\label{thm: mixed}
Let $\cF=\{(\lambda_{i},k_{i},d_{i})\}_{i=1}^{N}$ be a general family of triples with $\lambda_i$ and $d_i$ real and $k_i$ positive integers.
  Let $(\gamma_j)_{j=1}^{S'}$
  be the non-ascending sequence in which
  $\lambda_{i}$ occurs $k_{i}$ times.  $\cF$ is the family of
  normal parameters of a Gaussian state on $S$ modes iff the
  following conditions hold:
  \begin{itemize}
  \item[0.]  $\lambda_{i}> 0$ and $d_{i}\geq 0$.
  \item[1.] $S'$ is even, $S'=2S$.
  \item[2.]  $\gamma_{j} \gamma_{2S-j+1}\geq 1$.
  \end{itemize}
\end{theorem}

\begin{proof}
  Suppose first that the conditions hold. Let $\Gamma$ be the
  diagonal matrix with diagonal entries determined by
  $\Gamma_{2j-1}=\gamma_{j}$ and $\Gamma_{2j}=\gamma_{2S-j+1}$ for
  $j=1,\ldots,S$. As explained at the beginning of this section,
  $\Gamma$ is the covariance matrix of a Gaussian state where the
  $j$'th mode is a squeezed thermal state with temperature
  parameter $\nu_{j}=\sqrt{\gamma_{j}\gamma_{2S-j+1}}$ and
  squeezing parameter $r_{j}=\ln(\gamma_{j}/\nu_{j})/2$.  For the
  state to have the given normal parameters, it suffices to
  displace the quadratures associated with the first coordinate of
  each set of coordinates with identical eigenvalues by $d_{i}$.

  Let $\Gamma$ be the covariance matrix of a Gaussian state.
  Conditions 0. and 1. hold by the definition of the normal
  parameters. We prove that condition 2. holds. By
  Lem.~\ref{lem:exdiagcov}, there exists a diagonal covariance
  matrix $\Tau$ with the same spectrum as $\Gamma$. The diagonal
  block of $\Tau$ corresponding to mode $j$ has diagonal entries
  $\tau_{j}, \tau_{2S-j+1}$ and is the covariance matrix of a
  squeezed thermal state. This implies that
  $\tau_{j}\tau_{2S-j+1}\geq 1$.  By permuting the blocks and
  swapping the pair of quadrature coordinates in a block if necessary, we can assume that
  $\tau_{j}\geq \tau_{2S-j+1}$ and $\tau_{j}$ is non-ascending, which
  implies that the entire sequence $(\tau_{j})_{j=1}^{2S}$ is non-ascending.
  Since this sequence is the spectrum of $\Tau$ and $\Gamma$,
  it follows that $\gamma_{j}=\tau_{j}$ and condition 2. is satisfied. 
\end{proof}

We end this section with a brief discussion of the general
problem of characterizing the set of covariance matrices of Gaussian
states with a given family of normal parameters. We focus on the
case of no displacement, in which case the problem is to
characterize the set $\cG$ of covariance matrices $\Gamma$ of Gaussian states such that $\Gamma$ has a given spectrum, namely the
spectrum entailed by the family of normal parameters. Let $\Gamma_0$
be the diagonal covariance matrix of a Gaussian state with the same spectrum as $\Gamma$, constructed as in
the proof of Cor.~\ref{s4cor1}. Then $\cG$ is the intersection of
the orbit $\cO$ of $\Gamma_0$ under the orthogonal group $O(2S)$ and
the set $\cC$ of covariance matrices of Gaussian states. In general,
$\cG $ is a strict subset of $\cO$. For example, with $S=2$, the two
diagonal matrices $\diag(2,1/2, 4,1)$ and $\diag(1,1/2,4,2)$ are in
the same orbit of $O(2S)$, but the second one is not the covariance
matrix of a Gaussian state, because the first mode, associated with
the first two coordinates, violates the uncertainty principle, which
requires the product of the two diagonal entries to be at least $1$.

The set $\cG$ is a disjoint union of orbits under the group of
orthogonal and symplectic (OS) matrices. Each such orbit is identified
by its squeezing and its thermal spectrum. The results of this section
imply that for $S=1$ or for a pure state ($\det\Gamma_0 = 1$) $\cG$
consists of a single such orbit. In general, there are more orbits.
For example, if the spectrum of $\Gamma_0$ is $(4,3,1,1)$, then the
diagonal matrices $\diag(4,3,1,1)$ and $\diag(4,1,3,1)$ are both in
$\cG$ and have different squeezing spectra and thermal parameters.
Not all OS orbits have representatives that are diagonal.  Examples of
such orbits exist for $S\geq 2$.  Consider $S=2$. It suffices to
exhibit a covariance matrix $\Gamma$ that cannot be diagonalized by an
OS matrix. We construct $\Gamma$ such that its spectrum is
  different from that expected from its temperature and squeezing
  parameters. This prevents diagonalization by an OS matrix because if there exists an OS matrix $A$ such that $D=A\Gamma A^T$ is diagonal,
  then the diagonal of $D$ contains the spectrum and can be arranged to be of the form $(u_1
  e^{2r_1},u_1 e^{-2r_1}, u_2 e^{2r_2}, u_2 e^{-2r_2})$, where $u_1,
  u_2$ are the temperature parameters and $r_1, r_2$ are the squeezing
  parameters of $\Gamma$. To construct $\Gamma$, let $c,s,\tau$ be positive real numbers satisfying
$c^2 - s^2 = 1$ and define
\begin{align}
  \Delta &=
  \begin{pmatrix}
    1+2\tau&0&0&0\\
    0&1+2\tau&0&0\\
    0&0&1&0\\
    0&0&0&1
  \end{pmatrix}
  \nonumber\\
  R &=
  \begin{pmatrix}
    c&0&s&0\\
    0&c&0&-s\\
    s&0&c&0\\
    0&-s&0&c\\
  \end{pmatrix}
  \nonumber\\
  \Gamma &= R^T \Delta R.
\end{align}
Then $R$ is symplectic and $\Delta$ is a thermal, diagonal covariance
matrix. Therefore $\Gamma$ is the covariance matrix of a Gaussian
state. In this case the two temperature parameters are $u_1=1+2\tau$, $u_2=1$ and one can verify that the squeezing parameters are given by
$r_1 = r_2 = \frac{1}{2} \ln(\frac{c+s}{c-s})$ by computing $Q$ in
  the decomposition $R=KQL$, where $K$ and $L$ are OS and $Q$ is diagonal. Because $R$ is symmetric, $K=L^T$, so it suffices to check
  the spectrum of $R$. If $\Gamma$ were OS diagonalizable, because the squeezing and temperature parameters do not change under OS transformations, its
  spectrum would be $((1+2\tau)\frac{c+s}{c-s}$,
  $(1+2\tau)\frac{c-s}{c+s}$, $ \frac{c+s}{c-s}$,
  $\frac{c-s}{c+s})$. For $\tau>0$ and $s>0$, this consists of at
  least three distinct values.  However, direct calculation of the
  spectrum of $\Gamma$ shows that there are only two distinct
  eigenvalues $g_{\pm} = (1+\tau)(c^2 + s^2) \pm
  \sqrt{4(1+\tau)^2c^2s^2 + \tau^2}$, each with multiplicity $2$. We
  conclude that $\Gamma$ is not OS diagonalizable.
Because not all OS orbits have diagonal representatives, an analysis
of squeezing and thermal spectra of members of $\cG$ can not be
reduced to an analysis of diagonal covariance matrices with the given
spectrum.

\section{Conclusion\label{sec:conclusion}}

We investigated the problem of what properties of an arbitrary multimode Gaussian state are determined by the total photon number distribution. We
found that the total photon number distribution determines the spectrum
of the covariance matrix and the absolute displacement within each
eigenspace. For pure states this implies that the distribution determines the
squeezing parameters and the absolute displacement within each
subspace of the phase space where the Gaussian state has the same
amount of squeezing. For one mode in a mixed state, the temperature parameter can also be determined.
In general, we identified representatives for each equivalence
  class of Gaussian states with the same normal parameters and
  characterized the set of normal parameters of Gaussian states.

We established the mathematical relationship between photon number probabilities and the normal parameters consisting of the spectrum and the
displacement of a Gaussian state.  Since the number of normal
parameters is at most four times the number of modes, we expect that
it suffices to know a bounded number of photon number probabilities
to calculate the normal parameters. To make use of the mathematical
relationship in an experimental setting requires an effective way of
computing well-fitting normal parameters from a small number of estimated
photon number probabilities.

\begin{acknowledgments}
  This work includes contributions of the National Institute of
  Standards and Technology, which are not subject to U.S. copyright.
  A. Avagyan acknowledges support from the Professional Research Experience Program (PREP) operated jointly by NIST and the University of Colorado.
  We thank Alex Kwiatkowski and Yi-Kai Liu for helpful comments and suggestions to improve the paper, and Italo Pereira for informing us about relevant literature. Discussions with Thomas Gerrits and Krister Shalm served as a motivation for this study.
\end{acknowledgments}

\bibliography{ghom_gaussian_paper.bib}
\end{document}